\pdfoutput=1

\documentclass[10pt]{article}

\usepackage{amsfonts,amsmath,amssymb,amsthm}
\usepackage[pdftex]{graphicx}
\usepackage[hmargin=1in,vmargin=1in]{geometry}
\usepackage{natbib,setspace,enumerate,mathtools}
\usepackage{multirow}
\usepackage{booktabs}
\allowdisplaybreaks
\usepackage[toc,title,titletoc,header]{appendix}
\usepackage{adjustbox}
\usepackage[colorlinks,citecolor=blue]{hyperref}
\usepackage{etoolbox} 
\usepackage{threeparttable}
\def\qed{\rule{2mm}{2mm}}
\def\independent{\perp \!\!\! \perp}

\newtheorem{theorem}{Theorem}[section]

\theoremstyle{definition}

\newtheorem{remark}{Remark}[section]
\newtheorem{assumption}{Assumption}[section]

\AtEndEnvironment{remark}{~\qed}
\AtEndEnvironment{example}{~\qed}

\newcommand{\mycomment}[1]{}

\begin{document}

\author{
Jizhou Liu \\
Booth School of Business\\
University of Chicago\\
\url{jliu32@chicagobooth.edu}
\and
Eric J. Tchetgen Tchetgen \\
The Wharton School \\
University of Pennsylvania \\
\url{ett@wharton.upenn.edu}
\and 
Carlos Varjão \\
Senior Economist \\
Amazon.com, Inc \\
\url{varjaoc@amazon.com}
}

\bigskip

\title{Proximal Causal Inference for Synthetic Control with Surrogates \thanks{This paper has been generously funded by Amazon.com, Inc., through a graduate student fellowship gift awarded to Jizhou Liu.}}

\maketitle

\vspace{-0.3in}

\begin{spacing}{1.2}
\begin{abstract}
The synthetic control method (SCM) has become a popular tool for estimating causal effects in policy evaluation, where a single treated unit is observed, and a heterogeneous set of untreated units with pre- and post-policy change data are also observed. However, the synthetic control method faces challenges in accurately predicting post-intervention potential outcome had, contrary to fact, the treatment been withheld, when the pre-intervention period is short or the post-intervention period is long. To address these issues, we propose a novel method that leverages post-intervention information, specifically time-varying correlates of the causal effect called ``surrogates'', within the synthetic control framework. We establish conditions for identifying model parameters using the proximal inference framework and apply the generalized method of moments (GMM) approach for estimation and inference about the average treatment effect on the treated (ATT). Interestingly, we uncover specific conditions under which exclusively using post-intervention data suffices for estimation within our framework. Moreover, we explore several extensions, including covariates adjustment, relaxing linearity assumptions through non-parametric identification, and incorporating so-called ``contaminated'' surrogates, which do not exactly satisfy conditions to be valid surrogates but nevertheless can be incorporated via a simple modification of the proposed approach. Through a simulation study, we demonstrate that our method can outperform other synthetic control methods in estimating both short-term and long-term effects, yielding more accurate inferences. In an empirical application examining the Panic of 1907, one of the worst financial crises in U.S. history, we confirm the practical relevance of our theoretical results.
\end{abstract}
\end{spacing}
 

\hypersetup{pageanchor=false}
\thispagestyle{empty} 
\newpage
\hypersetup{pageanchor=true}
\setcounter{page}{1}

\section{Introduction}
The synthetic control method (SCM), first introduced by \cite{abadie2003}, has become a widely used baseline method in empirical research across economics and social sciences. It is designed to estimate the impact of an intervention experienced by a single unit by predicting the counterfactual scenario of what would have occurred if the intervention had not taken place, utilizing outcomes from unaffected units. As formalized in \cite{abadie2010}, the approach involves matching the treated unit to a weighted average of control units, referred to as "synthetic control," such that the resulting synthetic control matches the pre-intervention outcome trajectory over time for the treated units. Despite its significance in policy evaluation literature, the SCM may struggle to accurately reproduce the trajectory of the outcome for the treated unit when faced with a limited number of pre-intervention periods or structural breaks during a long pre- and post-intervention period (see \cite{abadie2021}). In response to these challenges, this paper proposes a novel method that utilizes post-intervention information to enhance the precision of the estimated causal effects.

Our approach leverages time-varying correlates of the causal effect that can be integrated into the synthetic control framework. Contrary to standard SCM, which requires donor units to remain unaffected by the intervention, our proposed model repurposes correlates as ideal surrogates of causal effects to the extent that they are predictive of treatment effects. This innovative method aims to enhance the performance of SCM in various situations, including estimating long-term effects and handling short pre-intervention periods, by incorporating additional post-intervention information, specifically surrogates.

The synthetic control literature has primarily focused on frameworks that assume either a perfect match of the synthetic control to the treated unit in the pre-treatment period (see for example \cite{abadie2010}) or a perfect match of the underlying unobserved factor loadings (see for example \cite{Ferman2019, Amjad2018, Powell2018, Ferman2021}). However, it is rarely the case that the perfect match assumption holds exactly in the observed data. Furthermore, models that assume perfect match of underlying factor loadings may be subject to bias due to the measurement error that arises in the synthetic control regression of the treated unit on donor outcomes. To address this issue, some researchers have proposed using lagged control unit outcomes as instrumental variables \cite{Ferman2021} or de-noising the data through singular value thresholding \cite{Amjad2018}, under an assumption of independent idiosyncratic error terms across units and time. In contrast, our paper employs the proximal causal inference framework proposed by \cite{PI2021}, which offers several advantages, such as effectively addressing the measurement error issue and providing convenient access to off-the-shelf Generalized Method of Moments (GMM) approaches for estimation and inference of average treatment effects on the treated (ATT).

We first establish the conditions required for identifying model parameters within the proximal inference framework. We then employ the generalized method of moments (GMM) approach to estimate and conduct statistical inference on the average treatment effect on the treated (ATT). In the pre-intervention period, we rely on the donor series to learn the ``synthetic control'', conceptually  similar to the classic synthetic control method (see \cite{PI2021}). In the post-intervention period, our method augments the estimation of synthetic weights and the causal effect series with surrogates. Our identification and estimation strategy is based on a joint set of moment equations that utilize data from both periods. Under this framework, we highlight the possibility of using only post-treatment data for estimation. This is potentially feasible as outcomes of treatment comprise of two components: one influenced by the latent factors affecting donors, and the other by those affecting surrogates. Because our method utilizes observations for both surrogates and donors, it allows for simultaneous learning of these components. 

Next, we explore several extensions to our proposed framework. First, we incorporate adjustments for measured covariates that are causally unaffected by the treatment. Second, we aim to relax the linear structure in our model, outlining sufficient conditions for nonparametric identification, which can guide practitioners in implementing nonparametric methods. Finally, we consider so-called "contaminated" surrogates, which, in addition to being time-varying correlates of the causal effect, also have components correlated with the outcomes of donors. Good surrogate candidates in practice can be defined as auxiliary outcomes of the target unit or outcomes of other units potentially causally affected by the intervention.

Finally, we compare our proposed synthetic control methodology to various methods proposed in previous literature, including conventional synthetic control methods and proximal causal inference, through a simulation study and empirical application. Our simulation study provides evidence that our method can outperform methods, with lower Mean Squared Error (MSE) and good coverage rates under a relative broad range of model specifications. We demonstrate that our method produces estimated causal effects that are more closely aligned with ground truth than the classic synthetic control regression method in the post-intervention period, particularly when the post-treatment period is long. Our study also evaluates the potential for applying proximal causal inference to estimating causal effects with post-intervention data only. In our empirical application, we explore the Panic of 1907, one of the most severe financial crises in US history, and assess its impact on NYC trusts using a high-frequency dataset, as analyzed by \cite{Fohlin2021}. We apply and compare various synthetic control methods to the stock prices of trust companies, demonstrating that the proximal causal inference methods, especially our proposed approach, produce more credible point estimates and inferences. Moreover, we highlight the potential of exclusively using post-intervention data for estimation within our framework.

Our method contributes to the expanding literature on methodological advancements that enhance traditional synthetic control methods by incorporating supplementary estimation techniques or additional information. \cite{feller2021} propose an augmented synthetic control method that employs ridge regression to rectify disparities in pre-treatment outcomes between treated units and un-augmented synthetic control estimators. \cite{Kellogg2021} advocate for the integration of synthetic control estimators and nearest-neighbor matching estimators in a weighted average to mitigate interpolation and extrapolation biases. A number of prior studies have utilized penalization schemes to regularize synthetic control weights (for example, \cite{Abadie2021jasa}, \cite{Doudchenko2016}, \cite{Chernozhukov2021}). Our research is also closely connected to recent advancements in the field of proximal causal inference (for instance, \cite{miao2018, Shi2020,Cui2023,miao2020}).

The rest of the paper is organized as follows. In Section \ref{sec:setup} we describe our setup and notation. Section \ref{sec:id-est} presents the main results for identification and estimation. In section \ref{sec:extensions}, we discuss a few extensions to the original framework. In Section \ref{sec:simulation}, we examine the finite sample behavior of various synthetic control methods through simulations. In Section \ref{sec:empirical} we illustrate our proposed method in an empirical application based on the data from \cite{Fohlin2021}. Finally, we close with a brief discussion in Section \ref{sec:conclusion}.

\section{Setup and Notation}\label{sec:setup}
Let $Y_{t}$ denote the observed outcome for the target unit, $W_{t}$ denote a $(N\times 1)$ vector of the observed outcome for donors, and $X_{t}$ denote a $(K\times 1)$ vector of observed surrogates for time periods $t=1,\dots, T$. Denote by $Y_{t}(0)$  the outcome of the target unit at time period $t$ that would be observed in the absence of the intervention, and by $Y_{t}(1)$ he outcome of the target unit at time period $t$ that would be observed in the presence of the intervention. The (observed) outcome and potential outcomes are related to treatment assignment by the relationship
\begin{equation}\label{eqn:outcome}
    Y_t = \mathbf{1}\left( t > T_0 \right) Y_t(1) + \mathbf{1}\left( t \leq T_0 \right) Y_t(0)~,
\end{equation}
where $1 < T_0 < T$ denotes the time period when the target unit receives treatment. In this paper, we focus on estimation and inference on the so-called average treatment effect on the treated (ATT), which is given by:
\begin{equation}\label{eqn:estimand}
    \tau = E\left[\frac{1}{T-T_0} \sum_{t=T_0}^{T} Y_t(1) - Y_t(0) \right]~.
\end{equation}

Suppose that $Y_{t}(0)$ and $W_t$ follow an interactive fixed effect model (IFEM).
\begin{equation}\label{eqn:ifem-donor}
\begin{split}
    &Y_t(0) = \lambda_t^\prime \beta + \epsilon_{Y,t} \\
    &W_t^\prime = \lambda_t^\prime \Gamma + \epsilon_{W,t}^\prime~,
\end{split}
\end{equation}
where $\lambda_{t}$ is a $(F\times 1)$ vector of common factors, for $t=1,\dots,T$, $\beta$ is an $(F\times 1)$ vector of factor loadings, $\Gamma$ is an $(F\times N)$ matrix of factor loadings, $\epsilon_{Y,t}$ and $\epsilon_{W,t}$ are the corresponding error terms. Similarly, we consider a synthetic surrogate model in post intervention period based on a similar factor model:
\begin{equation}\label{eqn:ifem-surrogate}
\begin{split}
    &Y_t(1) - Y_t(0) = \rho_t^\prime \theta + \delta_t \\
    & X_t^\prime = \rho_t^\prime \Phi + \epsilon^\prime_{X,t}~,
\end{split}
\end{equation}
where $\rho_{t}$ is a $(K\times 1)$ vector of common factors, for $t=1,\dots,T$, and $\theta$ is an $(K\times 1)$ vector of factor loadings, $X_t$ is an $(H\times 1)$ vector of surrogates, $\Phi$ is a $(K\times H)$ matrix of factor loadings, $\delta_t$ and $\epsilon_{X,t}$ are the corresponding error terms.

The IFEM model in Equation (\ref{eqn:ifem-donor}) is commonly assumed in the synthetic control literature \citep[see, for example,][]{xu2017, Ferman2021, PI2021}. In this paper, we consider an additional surrogate model, as in Equation (\ref{eqn:ifem-surrogate}), which is analogous to Equation (\ref{eqn:ifem-donor}). The motivation is to leverage extra post-intervention information sometimes available in the dataset to capture post-intervention period variation in causal effects, using surrogates that are highly correlated with post-intervention outcomes. This property is manifested in Equation (\ref{eqn:ifem-surrogate}), where both surrogates $X_t$ and causal effects $Y_t(1)-Y_t(0)$ are driven by the same set of latent factors, as assumed in the classical synthetic control approach in Equation (\ref{eqn:ifem-donor}). To identify and estimate parameters in the model, we need a few more assumptions. We begin by making the following assumptions on error terms.
\begin{assumption}[Conditional independent errors] \label{as:error-term}
The error terms are conditional mean independent and independent of each other. $E\left(\epsilon_{Y,t} \mid \lambda_t \right) = E\left(\epsilon_{W,t} \mid \lambda_t \right) = E\left(\delta_t \mid \rho_t \right) =  E\left(\epsilon_{X,t} \mid \rho_t \right) = 0$ for $t=1,\dots,T$.
\end{assumption}
Further, we make the following assumption that is key to identification of a synthetic control and the time series of post-intervention causal effects:
\begin{assumption}[Existence of synthetic control and surrogacy]\label{as:existence-of-weight}
There exist a $(N\times 1)$ vector of weights $\alpha$ such that $\Gamma \alpha = \beta$, and a $(H\times 1)$ vector of weights $\gamma$ such that $\Phi \gamma = \theta$.
\end{assumption}
\begin{remark}\label{remark:exist}
    A necessary condition for the existence of such a vector of weights $\alpha$ (or $\gamma$) is that the number of donors exceeds the number of latent factors $\lambda_t$ (or the number of surrogates exceeds the number of latent factors $\rho_t$), i.e. $N \geq F$ (or $H\geq K$).
\end{remark}
Next, we make the following assumption to identify synthetic weights $\alpha$ and coefficients of the surrogate $\gamma$, which can later be used for identifying causal effects:
\begin{assumption}[Existence of proxy]\label{as:existence-of-proxy}
We observe proxies $\{Z_{0,t}, Z_{1,t}\}$ such that
\begin{align*}
    Z_{0,t} &\independent \{Y_i(0), W_t\} \mid \lambda_t \text{ for any } t \leq T_0 ~, \\
    \{Z_{0,t}, Z_{1,t}\} &\independent \{Y_i(1), W_t, X_t\} \mid \lambda_t, \rho_t \text{ for any } t > T_0~.
\end{align*}
\end{assumption}
Assumption \ref{as:existence-of-proxy} implies that $\lambda_t$ encompasses the entire set of confounders between $Z_{0,t}$ and $\{Y_i(0), W_t\}$, while $(\lambda_t,\rho_t)$ encompasses the entire set of confounders between $\{Z_{0,t}, Z_{1,t}\}$ and $\{Y_i(1), W_t, X_t\}$. As argued in \cite{PI2021}, $Z_{0,t}$, which denotes proxy variables for outcomes of donors, is frequently available in synthetic control settings. For instance, a suitable choice of proxy variables might be the outcome of units that are potential donor candidates but were excluded from the donor pool. Similarly, proxy variables for surrogates, denoted by $Z_{1,t}$, could be potentially unused surrogate candidates.

In order to ensure unique identification of the synthetic control and surrogate weights, we introduce the following completeness condition.
\begin{assumption}[Completeness] \label{as:completeness}
    For any square integrable function $g$, if
    \begin{equation*}
        E[g(W_t) \mid Z_{0,t}=z_{0,t}, t \leq T_0] = 0\\
    \end{equation*}
    for all $z_{0,t}$, then $g(W_t)=0$ for any $t \leq T_0$.
    For any square integrable function $h$, if
    \begin{align*}
        E[h(W_t, X_t) \mid Z_{0,t}=z_{0,t}, Z_{1,t}=z_{1,t}, t > T_0] = 0
    \end{align*}
    for all $z_{0,t}, z_{1,t}$, then $h(W_t, X_t)=0$ for any $t > T_0$.
\end{assumption}
Assumption \ref{as:completeness} is formally known as completeness condition which characterizes the informativeness of $Z_{0,t}$ about $W_{t}$ (similarly for $\{Z_{0,t}, Z_{1,t}\}$ and $\{W_t, X_t\}$ in the post-intervention period), in the sense that any infinitesimal variation in $W_t$ is captured by variation in $Z_{0,t}$ such that no information has been lost through projection of $W_t$ on $Z_{0,t}$ in the pre-treatment period. 

Building on Assumptions \ref{as:existence-of-proxy}-\ref{as:completeness}, we present a practical approach for identifying and estimating causal effects by finding all necessary variables (donors, surrogates, and proxies). The initial step involves identifying as many appropriate donors and surrogates as possible to ensure $N > F$ and $H > K$. To find suitable donors, we adhere to the standard approach in synthetic control literature. For instance, in the German reunification example, \cite{abadie2021} considered a set of industrialized countries as potential donors for West Germany. Regarding surrogates, strong candidates should be highly predictive of magnitude of the causal effect for the target unit, which may include covariates that potentially mediate the causal effect of interest for the target unit in the post-intervention period and outcomes or covariates of other units affected by the treatment. Subsequently, we select $N_0$ donors to construct synthetic controls and $H_0$ surrogates to predict causal effects. The remaining $N - N_0$ donors and $H - H_0$ surrogates can serve as $Z_{0,t}$ and $Z_{1,t}$.

\section{Identification and Estimation}\label{sec:id-est}
In this section, we present our main results on identification and estimation. We apply the off-the-shelf GMM approach for identifying model parameters and estimating the average treatment effect on the treated (ATT). To begin with, under assumptions provided in the previous section, we have the following identification results:
\begin{theorem}[Identification]\label{thm:identification}
Under Assumption \ref{as:error-term}-\ref{as:existence-of-weight}, we have $E[Y_t(0)] = E[W_t^\prime \alpha]$ and $E[Y(1)] = E[W_t^\prime \alpha + X_t^\prime \gamma]$ for any $1\leq t \leq T$. Under Assumption \ref{as:error-term}-\ref{as:completeness}, $(\alpha, \gamma)$ is uniquely identified by solving 
\begin{align}
    E[Y_t - W_t^\prime \alpha \mid Z_{0,t}, t\leq T_0] &= 0 \text{ for any } t\leq T_0 ~,\\
    E[Y_t - W_t^\prime \alpha - X_t^\prime \gamma \mid Z_{0,t}, Z_{1,t}, t > T_0] &= 0 \text{ for any } t > T_0~.\label{eqn:moment-cond2}
\end{align}
\end{theorem}
Using the identification results, we formulate the following moment conditions:
\begin{align*}
    E[U_t(\alpha, \gamma)] = E\left[\begin{pmatrix}
    g_0(Z_{0,t}) (Y_t - W_t^\prime \alpha) I\{t \leq T_0\} \\
    g_1(Z_{0,t}, Z_{1,t}) (Y_t - W_t^\prime \alpha - X_t^\prime \gamma) I\{t > T_0\}
    \end{pmatrix}\right] = 0~,
\end{align*}
where $g_0, g_1$ are $d_0$- and $d_1$-dimensional ($d_0 + d_1 \geq N + H$) vector of user-specified functions. The above moment equation motivates the use of the generalized method of moments (GMM) method with $U_t(\alpha, \gamma)$ providing identifying moment restrictions (\cite{Hansen1982}). The GMM solves
\begin{equation*}
    (\hat{\alpha}, \hat{\gamma}) =\arg \min _{\alpha, \gamma} m\left(\alpha, \gamma\right)^{\prime} \Omega m\left(\alpha, \gamma\right),
\end{equation*}
where $m\left(\alpha, \gamma\right) = 1/T \sum_{t=1}^T U_t(\alpha, \gamma)$ is the sample moments evaluated at an arbitrary value $\alpha, \gamma$ and $\Omega$ is a $(d_0 + d_1)\times (d_0 + d_1)$ user-specified symmetric and positive-definite weight matrix.

An interesting observation is that by forming the unconditional moments $E[g_1(Z_{0,t}, Z_{1,t}) (Y_t - W_t^\prime \alpha - X_t^\prime \gamma) I\{t > T_0\}]=0$ using equation (\ref{eqn:moment-cond2}), we can identify $\alpha$ and $\gamma$, as long as the dimension of $g_1$ exceeds the number of donors and surrogates. This finding suggests the possibility of using only post-treatment data for estimation. In Section \ref{sec:simulation}, we validate the practical relevance of this idea through a simulation study.

Then, we make the following assumptions on the error terms:
\begin{assumption}[Stationary error]\label{as:stationary}
    $\epsilon_{Y,t}, \epsilon_{W,t}, \delta_t, \epsilon_{X,t}$ are stationary and weakly dependent processes. 
\end{assumption}
We adopt the weakly dependent assumption as a substitute for the independent and identically distributed errors assumption, while still ensuring the applicability of ergodic laws of large numbers and central limit theorems. The consistency and asymptotic normality of GMM estimators have been well-documented under such conditions, as demonstrated in studies by \cite{Hansen1982}, \cite{NeweyWest1986}, and \cite{Hall2005}. The theorem presented below is derived from Theorem 3.2 in \cite{Hall2005} and its proof can be found in Section B of the Supplementary Material.

\begin{theorem}\label{thm:gmm-coefficients}
    Under Assumptions \ref{as:error-term}-\ref{as:completeness} and \ref{as:stationary}  and regularity conditions \ref{as:app-stationary}-\ref{as:app-Gt} listed in Section \ref{sec:proof2}, as $T\rightarrow \infty$, we have
    \begin{equation*}
        \sqrt{T} (\hat\theta - \theta) \xrightarrow{d} N\left(0, \left(G^{\top} \Omega G\right)^{-1} G^{\top} \Omega S \Omega^{\top} G\left(G^{\top} \Omega^{\top} G\right)^{-1}\right) 
    \end{equation*}
    where $\theta = (\alpha^\prime, \gamma^\prime)^\prime$, $S = \lim_{T\rightarrow \infty} \operatorname{Var}\{\sqrt{T}m(\alpha,\gamma)\}$ is the variance-covariance matrix of the limiting distribution of $\sqrt{T} m(\alpha,\gamma)$ and
    \begin{equation*}
        G = \begin{pmatrix}
            g_0(Z_{0,t}) W_t^\prime I\{t\leq T_0\} & 0 \\
            g_1(Z_{0,t}, Z_{1,t}) W_t^\prime I\{t> T_0\} & g_1(Z_{0,t}, Z_{1,t}) X_t^\prime I\{t>T_0\}
        \end{pmatrix}~.
    \end{equation*}
\end{theorem}
\begin{remark}
    When utilizing the general method of moments framework, over-identification tests on moment restrictions are applicable when the number of moment conditions surpasses the number of parameters to be estimated (i.e., $d_0 + d_1 > N + H$). A testing procedure for this scenario was proposed by \cite{Hansen1982}, which involves calculating the $J$-statistic as follows:
    \begin{equation}\label{eqn:J-test}
        J_T = m(\hat\theta)^\prime \hat\Omega m(\hat\theta).
    \end{equation}
    The validity of this testing procedure is supported by its asymptotic distribution, which is a chi-square distribution with $d_0 + d_1 - N + H$ degrees of freedom, as demonstrated in Lemma 4.2 of \cite{Hansen1982}.
\end{remark}
\begin{remark}
    It is worth noting that Assumption \ref{as:completeness} is a stronger condition than what is strictly necessary for estimating the parameters of interest, $\theta$. In the context of the IFEM model, we only require the invertibility of $G^{\top} \Omega G$. In other words, the rank conditions ensuring the relevance of $Z_{0,t}$ and $W_t$ (as well as ${Z_{0,t}, Z_{1,t}}$ and ${W_t, X_t}$) are sufficient for the validity of the result.
\end{remark}
Theorem \ref{thm:gmm-coefficients} provides the asymptotic results required for estimation and statistical inference of the parameters $(\alpha, \gamma)$ identified by the structural model. The synthetic weights $\alpha$ potentially provide meaningful interpretation for the synthetic control, while $\gamma$ can heuristically be used to gauge the potential predictive ability of surrogates for the target treatment effect over time. We aim to make inferences about the ATT, for which we provide two distinct estimation strategies that we incorporate in the GMM framework. Note that, we have
\begin{equation*}
    \tau = E\left[\frac{1}{T-T_0} \sum_{t=T_0}^{T} Y_t(1) - Y_t(0) \right] = E\left[\frac{1}{T-T_0} \sum_{t=T_0}^{T} \rho_t^\prime\theta \right] = E\left[\frac{1}{T-T_0} \sum_{t=T_0}^{T} X_t^\prime \gamma \right]~,
\end{equation*}
which motivates the following estimator that plugs in the estimated surrogate coefficients and replaces population quantities with their sample counterparts:
\begin{equation*}
    \hat\tau = \frac{1}{T-T_0} \sum_{t=T_0}^{T} X_t^\prime \hat\gamma~.
\end{equation*}
Alternatively, we can use the difference between outcomes of the target unit and the synthetic control to estimate the average treatment effect as follows:
\begin{equation*}
    \tau = E\left[\frac{1}{T-T_0} \sum_{t=T_0}^{T} Y_t(1) - W_t^\prime \alpha \right]~,
\end{equation*}
which motivates the following plug-in estimator:
\begin{equation*}
    \hat\tau = \frac{1}{T-T_0} \sum_{t=T_0}^{T} Y_t -  W_t^\prime \hat\alpha~.
\end{equation*}
By incorporating both estimation strategies, we form the moment conditions $E[\Tilde{U}_t(\alpha, \gamma, \tau)] = 0$ where
\begin{equation*}
    \Tilde{U}_t(\alpha, \gamma, \tau) = \begin{pmatrix}
    g_0(Z_{0,t}) (Y_t - W_t^\prime \alpha) I\{t \leq T_0\} \\
    g_1(Z_{0,t},Z_{1,t}) (Y_t - W_t^\prime \alpha - X_t^\prime \gamma) I\{t > T_0\} \\
    (Y_t - \tau - W_t^\prime \alpha) I\{t > T_0\} \\
    (X_t^\prime \gamma - \tau)I\{t > T_0\}
    \end{pmatrix}~.
\end{equation*}
Then, we can estimate $\tau$ and structural paramters $(\alpha,\gamma)$ simultaneously through the GMM approach that solves
\begin{equation*}
    (\hat{\alpha}, \hat{\gamma}, \tau) =\arg \min _{\alpha, \gamma, \tau} \Tilde m\left(\alpha, \gamma, \tau\right)^{\prime} \Omega \Tilde m\left(\alpha, \gamma, \tau\right)
\end{equation*}
where $\Tilde m\left(\alpha, \gamma,\tau\right) = 1/T \sum_{t=1}^T \Tilde{U}_t(\alpha, \gamma, \tau)$. Similar to that of Theorem \ref{thm:gmm-coefficients}, we then present the following theorem for estimation and inference of ATT.
\begin{theorem}\label{thm:gmm-ate}
    Under Assumptions \ref{as:error-term}-\ref{as:completeness} and \ref{as:stationary} and regularity conditions \ref{as:app-stationary}-\ref{as:app-Gt} listed in Section \ref{sec:proof2}, as $T\rightarrow \infty$, we have
    \begin{equation*}
        \sqrt{T} (\hat\theta - \theta)
         \xrightarrow{d} N\left(0, \left(G^{\top} \Omega G\right)^{-1} G^{\top} \Omega S \Omega^{\top} G\left(G^{\top} \Omega^{\top} G\right)^{-1}\right) 
    \end{equation*}
    where $\theta = (\alpha^\prime, \gamma^\prime, \tau^\prime)^\prime$, $S = \lim_{T\rightarrow \infty} \operatorname{Var}\{\sqrt{T} \Tilde m(\alpha,\gamma, \tau)\}$ is the variance-covariance matrix of the limiting distribution of $\sqrt{T} \Tilde m(\alpha,\gamma, \tau)$ and
    \begin{equation*}
        G = E\left[\begin{pmatrix}
            g_0(Z_{0,t}) W_t^\prime I\{t\leq T_0\} & 0 & 0\\
            g_1(Z_{0,t},Z_{1,t}) W_t^\prime I\{t> T_0\} & g_1(Z_{0,t},Z_{1,t}) X_t^\prime I\{t>T_0\} & 0 \\
            W_t^\prime I\{t > T_0\} & 0 & I\{t > T_0\} \\
            0 & -X_t^\prime I\{t > T_0\} & I\{t > T_0\}
        \end{pmatrix}\right]~.
    \end{equation*}
\end{theorem}
\begin{remark}
    The GMM approach provided in Theorem \ref{thm:gmm-ate} offers a flexible framework for conducting estimation and inference for various types of causal effects. For example, the moment equations $\Tilde{U}_t(\alpha, \gamma, \tau)$ can be modified for estimating ATT of any specified period of time with $E[(X_t^\prime \gamma - \tau)I\{t_1 < t < t_2\}] = 0$, percentage lift with $E[(X_t^\prime \gamma -  W_t^\prime \alpha \cdot \tau )I\{t_1 < t < t_2\}]$, and so on.
\end{remark}

\section{Extensions}\label{sec:extensions}

In this section, we explore three extensions to our primary framework. First, we incorporate adjustments for measured covariates, which can be beneficial in practice to mitigate the issue of unmeasured confounders. Second, we aim to relax the linear structure in the IFEM model. We outline the necessary assumptions for nonparametric identification, which may be useful for practitioners when designing and implementing nonparametric methods in practice. Finally, we include a useful extension of the surrogacy model to consider ``contaminated'' surrogates. Such surrogates can potentially be alternative outcomes of the target unit or outcomes of other units affected by the intervention. 

\subsection{Adjustment for Measured Covariates}
In practice, one may wish to incorporate available covariate data measured across units and over time, either to account for endogeneity or to improve efficiency. As such, we generalize the IFEM model as follows:
\begin{align*}
    &Y_t(1) = \rho_t^\prime \theta + \delta_t  + \lambda_t^\prime \beta + C_{Y,t}^{\prime} \xi_{Y} + \epsilon_{Y,t} \\
    &Y_t(0) = \lambda_t^\prime \beta + C_{Y,t}^{\prime} \xi_{Y} + \epsilon_{Y,t} \\
    &W_t^\prime = \lambda_t^\prime \Gamma + \xi_W^\prime C_{W,t} +  \epsilon_{W,t}^\prime \\
    & X_t^\prime = \rho_t^\prime \Phi + \xi_{X}^\prime C_{X,t} + \epsilon^\prime_{X,t}~,
\end{align*}
where $C_{Y,t}, C_{W,t}, C_{X,t}$ are vectors/matrices of measured covariates that are not causally impacted by the intervention, i.e. $C_{Y,t}(1) = C_{Y,t}(0) = C_{Y,t}$. We modify Assumption \ref{as:error-term} and \ref{as:existence-of-proxy} to include measured covariates as follows:
\begin{assumption}[Conditional independent errors] \label{as:app-error-term}
The error terms are conditional mean independent and independent of each other. $E\left(\epsilon_{Y,t} \mid \lambda_t, C_{Y,t} \right) = E\left(\epsilon_{W,t} \mid \lambda_t, C_{W,t} \right) = E\left(\delta_t \mid \rho_t \right) =  E\left(\epsilon_{X,t} \mid \rho_t, C_{X,t} \right) = 0$ for $t=1,\dots,T$.
\end{assumption}
\begin{assumption}[Existence of proxies]\label{as:app-existence-of-proxy}
We observe $\{Z_{0,t}, Z_{1,t}\}$ such that
\begin{align*}
    Z_{0,t} &\independent \{Y_i(0), W_t\} \mid \lambda_t, C_{Y,t}, C_{W,t} \text{ for any } t \leq T_0 ~, \\
    \{Z_{0,t}, Z_{1,t}\} &\independent \{Y_i(1), W_t, X_t\} \mid \lambda_t, \rho_t, C_{Y,t}, C_{W,t}, C_{X,t} \text{ for any } t > T_0~.
\end{align*}
\end{assumption}
Let $\Tilde{Y}_t = Y_t - C_{Y,t}^\prime \xi_Y$, $\Tilde{W}_{t} = W_t - C_{W,t}^\prime \xi_W$ and $\Tilde{X}_t = X_t - C_{X,t}^\prime \xi_X$. Under the modified assumptions above, we have the following identification results:
\begin{theorem}[Identification]\label{thm:app-identification}
Under Assumption \ref{as:app-error-term} and \ref{as:existence-of-weight}, we have $E[Y_t(0)] = E[C_{Y,t}^\prime \xi_Y+\Tilde{W}_t^\prime \alpha]$ and $E[Y(1)] = E[C_{Y,t}^\prime \xi_Y+ \Tilde W_t^\prime \alpha + \Tilde X_t^\prime \gamma]$ for any $1\leq t \leq T$. Under Assumption \ref{as:app-error-term}-\ref{as:app-existence-of-proxy}, \ref{as:existence-of-weight} and \ref{as:completeness}, $(\alpha, \gamma)$ is uniquely identified by solving 
\begin{align*}
    E[\Tilde Y_t - \Tilde W_t^\prime \alpha \mid Z_{0,t}, C_{Y,t}, C_{W,t}, t\leq T_0] &= 0 \text{ for any } t\leq T_0 ~,\\
    E[\Tilde Y_t - \Tilde W_t^\prime \alpha - \Tilde X_t^\prime \gamma \mid Z_{0,t}, Z_{1,t}, C_{Y,t}, C_{W,t}, C_{X,t}, t > T_0] &= 0 \text{ for any } t > T_0~.
\end{align*}
\end{theorem}
\begin{proof}
The proof is similar to that of Theorem \ref{thm:identification} in Appendix \ref{sec:proof}.
\end{proof}
Next, we present an example of applying Theorem \ref{thm:app-identification} to treatment effect estimations. First, let $\xi=(\xi_Y, \xi_W, \xi_X)$, $\Tilde Z_{0,t}=(Z_{0,t}, C_{Y,t}, C_{W,t})$. We have moment conditions $E[\Tilde{U}_t(\alpha, \gamma, \tau, \xi)] = 0$ where
\begin{equation*}
    \Tilde{U}_t(\alpha, \gamma, \tau, \xi) = \begin{pmatrix}
    g_0(\Tilde Z_{0,t}) (Y_t  - W_t^\prime \alpha - C_{Y,t}^\prime \xi_Y - C_{W,t}^\prime \xi_{W}  \alpha) I\{t \leq T_0\} \\
    g_1( \Tilde Z_{0,t},Z_{1,t}, C_{X,t}) (Y_t - W_t^\prime \alpha - X_t^\prime \gamma - C_{Y,t}^\prime \xi_Y - C_{W,t}^\prime \xi_{W}  \alpha - C_{X,t}^\prime \xi_X \gamma) I\{t > T_0\} \\
    (Y_t - \tau - W_t^\prime \alpha - C_{Y,t}^\prime \xi_Y - C_{W,t}^\prime \xi_{W}  \alpha) I\{t > T_0\} \\
    (X_t^\prime \gamma - C_{X,t}^\prime \xi_X \gamma - \tau)I\{t > T_0\}
    \end{pmatrix}~.
\end{equation*}
Lastly, we can employ the standard GMM approach as outlined in Section \ref{sec:id-est} for estimation and inference.

\subsection{Nonparametric Identification and Estimation}

The given approach in section \ref{sec:id-est} is heavily dependent on the linearity assumptions in equations (\ref{eqn:ifem-donor}) and (\ref{eqn:ifem-surrogate}). In practice, this linear structure could potentially be violated. Thus, in this section, we extend our proposed framework to obtain nonparametric identification and estimation. To begin, we replace Assumption \ref{as:existence-of-proxy} with the following assumption:
\begin{assumption}[Existence of confounding bridge]\label{as:nonparam}
    There exists a function $h(W_{t})$ such that the outcome model for $Y_t(0)$ is equivalent to a model for $h(W_t)$:
    \begin{equation*}
        E[Y_t(0) \mid \lambda_t] = E[h(W_t)\mid \lambda_t], \quad \forall t \geq 1.
    \end{equation*}
    Meanwhile, there exists a function $g(X_t)$ such that the causal effect model for $Y_t(1) - Y_t(0)$ is equivalent to a model for $g(X_t)$:
    \begin{equation*}
        E[Y_t(1) - Y_t(0) \mid \rho_t] = E[g(X_t)\mid \rho_t], \quad \forall t \geq 1.
    \end{equation*}
\end{assumption}
From Assumption \ref{as:nonparam}, we can immediately deduce that nonparametric identification of ATT can be expressed as follows:

\begin{equation*}
\tau = E[Y_t(1) - Y_t(0) \mid t > T_0 ] = E[g(X_t) \mid t > T_0 ].
\end{equation*}
Subsequently, the following theorem presents the nonparametric identification results for $h(\cdot), g(\cdot)$:
\begin{theorem}[Moment condition for $h(\cdot), g(\cdot)$]\label{thm:nonparam}
    Under Assumption \ref{as:error-term}, \ref{as:nonparam} and \ref{as:existence-of-proxy}, the confounding bridge function $h(\cdot), g(\cdot)$ satisfies the moment condition
    \begin{align*}
        E[Y_t - h(W_t) \mid Z_{0,t}, t\leq T_0] &= 0 \text{ for any } t\leq T_0 ~,\\
        E[Y_t - h(W_t) - g(X_t) \mid Z_{0,t}, Z_{1,t}, t > T_0] &= 0 \text{ for any } t > T_0~.
    \end{align*}
\end{theorem}
\begin{proof}
The proof follows the same procedure as in Appendix \ref{sec:proof}. By Assumption \ref{as:existence-of-proxy}, we have $E[Y_t - h(W_t) \mid Z_{0,t}, \lambda_t, t\leq T_0] = E[Y_t - h(W_t) \mid \lambda_t, t\leq T_0] = 0$. By integrating over $\lambda_t$, we obtain $E[Y_t - h(W_t) \mid Z_{0,t}, t\leq T_0]  = 0$. The same argument applies to the second moment condition, and thus the result follows.
\end{proof}

\subsection{Contaminated Surrogates}\label{sec:contaminate}
In this section, we extend our IFEM model to accommodate situations where ``pure'' surrogates are not available. Instead, we only observe surrogates that are ``contaminated'' by latent factors that drive the outcomes of the control. In many applications, such contaminated surrogates are essentially alternative outcomes of the target units and are expected to be driven by $\lambda_t$ in (\ref{eqn:ifem-donor}) in the pre-intervention periods, and driven by both $\lambda_t$ and $\rho_t$ in (\ref{eqn:ifem-surrogate}) in the post-intervention periods. For example, if the target outcome of interest is the GDP growth of California, contaminated surrogates can potentially be inflation and the unemployment rate of California, or the GDP growth, inflation, and unemployment rate of other states affected by the treatment.

To formalize this setup, we consider the following latent factor model for the surrogates as a replacement of $X_t$ from equation (\ref{eqn:ifem-surrogate}):
\begin{equation}\label{eqn:ifem2}
\begin{split}
    & X_t^\prime(0) = \lambda_t^\prime \Theta + \epsilon^\prime_{X,t} \\
    & X_t^\prime(1) - X_t^\prime(0) = \rho_t^\prime \Phi + \delta_{X,t} ~,
\end{split}
\end{equation}
where $\Theta$ is a ($F\times H$) matrix of factor loadings. In addition to Assumption \ref{as:app-existence-of-proxy}, we assume the following:
\begin{assumption}[Existence of synthetic control for contaminated surrogates]\label{as:existence-of-weight2}
There exist a $(N\times H)$ matrix of weights $\Psi$ such that $\Gamma \Psi = \Theta$.
\end{assumption}
\begin{remark}
    Similar to Remark \ref{remark:exist}, such an assumption holds when the number of donors exceeds the number of latent factors $\lambda_t$, i.e. $N \geq F$.
\end{remark}
Under the model of contaminated surrogates and the additional Assumption \ref{as:existence-of-weight2}, we have the following results for identification:
\begin{theorem}[Identification]\label{thm:identification_contaminate}
Under Assumption \ref{as:error-term}-\ref{as:existence-of-weight} and \ref{as:existence-of-weight2}, we have $E[Y_t(0)] = E[W_t^\prime \alpha]$, $E[X_t(0)^\prime] = E[W_t^\prime \Psi]$ and $E[Y(1)] = E[W_t^\prime \alpha + (X_t^\prime(1) - W_t^\prime \Psi) \gamma]$ for any $1\leq t \leq T$. Under Assumption \ref{as:error-term}-\ref{as:completeness}, $(\alpha, \gamma)$ is uniquely identified by solving 
\begin{align*}
    E[Y_t - W_t^\prime \alpha \mid Z_{0,t}, t\leq T_0] &= 0 \text{ for any } t\leq T_0 ~,\\
    E[X_t - W_t^\prime \Psi \mid Z_{0,t}, t\leq T_0] &= 0 \text{ for any } t\leq T_0 ~,\\
    E[Y_t - W_t^\prime \alpha - (X_t^\prime -  W_t^\prime \Psi)\gamma \mid Z_{0,t}, Z_{1,t}, t > T_0] &= 0 \text{ for any } t > T_0~.
\end{align*}
\end{theorem}
\begin{proof}
The proof is similar to that of Theorem \ref{thm:identification} in Appendix \ref{sec:proof}.
\end{proof}
Next, we use GMM approach for estimation and inference about the average treatment effects on the treated (ATT). The moment equations defining the ATT can be given as $E[(Y_t - \tau - W_t^\prime \alpha)I\{t > T_0\}] = 0$ and $E[((X_t^\prime - W_t^\prime \Psi)\gamma - \tau)I\{t > T_0\}] = 0$.

\section{Simulations}\label{sec:simulation}
In this section, we examine the finite sample performance of our proposed method under various conditions. We simulate time series data for $N$ control units and one treated unit with fixed pre-treatment period $T_0=100$ across different lengths of total periods $T=200, 400, 800$. We assume the following data generating mechanism:
\begin{align*}
    Y_t &= I\{t > T_0\}(\rho_t^\prime \theta + \delta_t) +\lambda_t^{\prime} \beta +C_{Y,t}^{\prime} \xi +\epsilon_{Y, t} \\
    W_{i,t} &=  \lambda_t^\prime \Gamma_i+C_{W,i,t}^{\prime} \xi+\epsilon_{i,t} \\
    X_{i,t} &=  \rho_t^\prime \Phi_{i} + C_{X,i,t}^{\prime} \xi + \epsilon^\prime_{X,t}
\end{align*}
where $\xi=1$ or $0$ corresponds to scenarios with or without measured covariates, $C_{Y,t}, C_{W,i,t}, C_{X,i,t} \stackrel{i . i . d}{\sim} N(0, 1)$ denotes a measured covariate, and $\epsilon_{Y, t}, \epsilon_{i,t}, \epsilon_{X,i,t} \stackrel{i . i . d}{\sim} N(0,1)$ denotes random errors. We simulate a vector of latent factors $\lambda_t = (\lambda_{1,t}, \dots, \lambda_{F,t})^\prime$ for donor model, where $\lambda_{k,t} \stackrel{i . i . d}{\sim} N(1, 1) \text{ or } \lambda_{k,t} \stackrel{i . i . d}{\sim} N(\text{log}(t), 1)$, $k=1,\dots, F$, and $F=1,5,10$, and a vector of latent factors $\rho_t = \{\rho_{1,t},\dots, \rho_{K,t}\}$ for surrogate model, where $\rho_{i,t} \stackrel{i . i . d}{\sim} N(\mu_i, 1)$, $j=1,\dots, K$, and $K=1,5,10$. We set $\mu_1 = 1$ and $\mu_{i} = 0$ for $2\leq i\leq K$. In this case, the average treatment effect is equal to one, i.e. $\tau = 1$. For each setting, we assume the number of control units $N = 2F$ and the number of surrogates $H = 2 K$, and the first half of the control units and surrogates constitute the donor and surrogate pool. We specify factor loadings as follows: $\beta = (1,\dots, 1)^\prime$ and $\theta = (1,\dots, 1)^\prime$ are two vectors of ones, and $\Gamma=(I_{F},I_{F}) \Phi=(I_{K},I_{K})$ stack two identity matrices, which correspondingly generate donors/surrogates and proxies. Respectively, $\Phi_i$ and $\Gamma_i$ correpsond to $i$-th column of $\Phi$ and $\Gamma$ . This specification leads to exact identified linear systems with solutions $\beta = \alpha$ and $\gamma = \theta$. Notably, the SC weights $\alpha$ do not sum to one.

To estimate $\tau$, we implement our proposed method taking the first half of control units as donors $\{W_{i,t}: i = 1,\dots, F\}$ and second half of control units as supplemental proxies $\{W_{i,t}: i = F+1,\dots, 2F\}$. Similarly, we use the first half of surrogates $\{X_{i,t}: i = 1,\dots, K\}$ to predict causal effects and the second half of surrogates as proxies $\{X_{i,t}: i = K+1,\dots, 2 K\}$. When there exists a measured covariate (i.e., $\xi=1$) which is predictive of the outcome, we implement our method with and without covariate adjustment to investigate whether there is an efficiency gain from adjusting for such a predictor of the outcome. As a comparison, we implement the following methods:
\begin{enumerate}
    \item \textbf{(SC)} The unconstrained OLS under following linear regression model: $Y_t = \alpha_0 + I\{t>T_0\} \tau + W_{t}^\prime \alpha  + \nu_t$, $t=1,\dots, T$. 
    \item \textbf{(SC-S)} The unconstrained OLS with surrogates under following linear regression model: $Y_t = \alpha_0 + I\{t>T_0\} X_{t}^\prime \gamma + W_{t}^\prime \alpha  + \nu_t$, $t=1,\dots, T$. We perform inference through the following moment equations:
    \begin{equation*}
        \Tilde{U}_t(\alpha, \gamma, \tau) = \begin{pmatrix}
        I\{t>T_0\} X_{t} (Y_t - (I\{t>T_0\} X_{t}^\prime \gamma + W_{t}^\prime \alpha  )\\
        W_{t} (Y_t - (I\{t>T_0\} X_{t}^\prime \gamma + W_{t}^\prime \alpha ) )\\
        C_{Y,t} (Y_t - (I\{t>T_0\} X_{t}^\prime \gamma + W_{t}^\prime \alpha ) )\\
        (X_t^\prime \gamma - \tau)I\{t > T_0\}
    \end{pmatrix}~.
    \end{equation*}
    \item \textbf{(PI)} The proximal inference method without using surrogates in \cite{PI2021}, which is based on the following moment conditions:
    \begin{equation*}
        \Tilde{U}_t(\alpha, \gamma, \tau) = \begin{pmatrix}
        Z_{0,t} (Y_t - W_t^\prime \alpha) I\{t \leq T_0\} \\
        (Y_t - \tau - W_t^\prime \alpha) I\{t > T_0\}
    \end{pmatrix}~.
    \end{equation*}
    \item \textbf{(PI-P)} A proximal inference method using post-treatment period only. Specifically, we use the following moment conditions:
    \begin{equation*}
        \Tilde{U}_t(\alpha, \gamma, \tau) = \begin{pmatrix}
        Z_{0,t} (Y_t - W_t^\prime \alpha - X_t^\prime \gamma) I\{t > T_0\} \\
        Z_{1,t} (Y_t - W_t^\prime \alpha - X_t^\prime \gamma) I\{t > T_0\} \\
        (X_t^\prime \gamma - \tau)I\{t > T_0\}
    \end{pmatrix}~.
    \end{equation*}
    \item \textbf{(PI-S)} Our proposed proximal inference with surrogates from Section \ref{sec:id-est}. Specifically, we use the following moment conditions:
    \begin{equation*}
        \Tilde{U}_t(\alpha, \gamma, \tau) = \begin{pmatrix}
        Z_{0,t} (Y_t - W_t^\prime \alpha) I\{t \leq T_0\} \\
        Z_{1,t} (Y_t - W_t^\prime \alpha - X_t^\prime \gamma) I\{t > T_0\} \\
        (X_t^\prime \gamma - \tau)I\{t > T_0\}
    \end{pmatrix}~.
    \end{equation*}
\end{enumerate}

In Figure \ref{fig:simulation}, we visualize the performance of the classic synthetic control regression \textbf{SC} and our approach \textbf{PI-S} with time series plots of estimated causal effects, plotted against the ground truth. Noticeably, both \textbf{SC} and \textbf{PI-S} move closely together in the pre-intervention period but diverge increasingly during the post-intervention period. It is evident that \textbf{PI-S} produces an estimated causal effect series that is more closely aligned with the ground truth than \textbf{SC} in the post-intervention period.

\begin{figure}[ht!]
    \centering
    \includegraphics[width=0.7\textwidth]{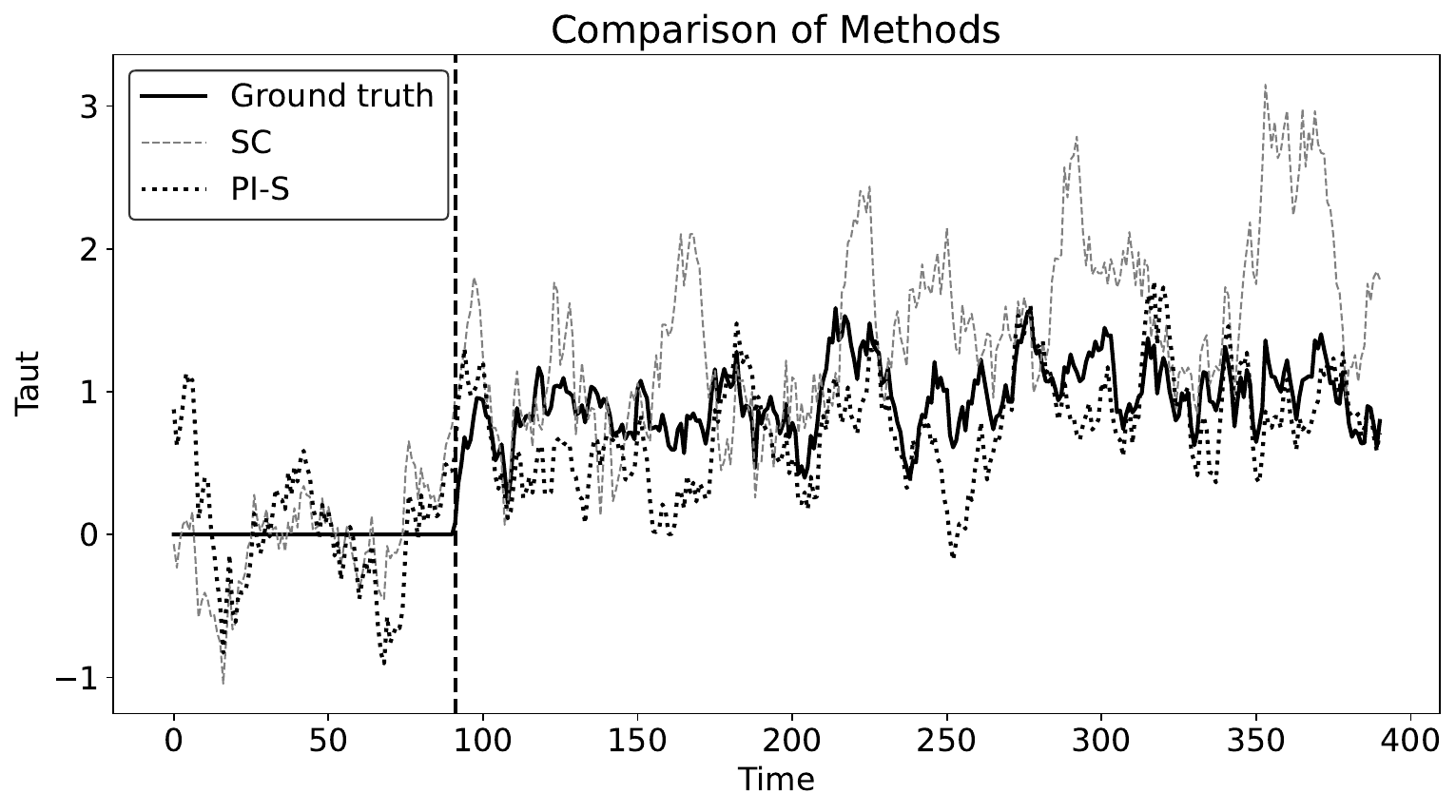}
    \caption{Time series of causal effects under \textbf{SC} and \textbf{PI-S}}
    \label{fig:simulation}
\end{figure}

Table \ref{table:mse} reports the Mean Squared Errors (MSE) of each estimation method for both independent and AR(1) errors, computed across 2,000 Monte Carlo iterations. We compare our proposed method, \textbf{PI-S}, with other methods under two data generating processes, one with an i.i.d normal latent factor $\lambda_t$ with a mean and variance equal to one, and another with an i.i.d normal $\lambda_t$ with a time trend $\log(t)$. We find that \textbf{PI-S} outperforms all other methods under almost all model specifications, having the lowest MSE. However, under $\lambda_t \sim N(1,1)$, \textbf{SC} and \textbf{SC-S} perform better than \textbf{PI-S} when the number of factors is large relative to the number of total time periods, i.e. $K=5, T=200$, possibly due to the "curse of dimensionality." Additionally, we observe that the performance of \textbf{PI-P} and \textbf{PI-S} improves as the post-treatment period increases, while the performance of other methods deteriorates, highlighting the advantage of our approach in the case of long post-treatment periods. Finally, we also find that as the dimension of latent factors increases, the performance of all proximal inference methods declines, particularly for \textbf{PI-P}.

\begin{table}[ht!]
\centering
\setlength{\tabcolsep}{5.5pt}
\begin{adjustbox}{max width=0.95\linewidth,center}
\begin{tabular}{lllllllllllll}
\toprule
       &              & \multicolumn{5}{c}{$\lambda_t \sim N(1,1)$}                      &  & \multicolumn{5}{c}{$\lambda_t \sim N(\log(t),1)$}                      \\  \cmidrule{3-7} \cmidrule{9-13}
Std Error & ($K,T$)  &  \textbf{SC}   & \textbf{SC-S}   & \textbf{PI} & \textbf{PI-P}  & \textbf{PI-S}   & & \textbf{SC}   & \textbf{SC-S}   & \textbf{PI} & \textbf{PI-P}  & \textbf{PI-S}   \\
\toprule
\multirow{6}{*}{Robust} & $(1,200)$  & 0.057 & 0.066 & 0.060 & 0.073 & 0.048 &  & 0.684  & 0.690  & 0.078 & 0.115 & 0.051 \\
                        & $(1,400)$  & 0.030 & 0.033 & 0.031 & 0.022 & 0.018 &  & 1.232  & 1.236  & 0.056 & 0.034 & 0.025 \\
                        & $(1,800)$  & 0.022 & 0.024 & 0.023 & 0.009 & 0.010 &  & 1.961  & 1.963  & 0.061 & 0.014 & 0.020 \\
                        & $(5,200)$  & 0.150 & 0.158 & 0.215 & 2.617 & 0.201 &  & 6.886  & 6.749  & 0.290 & 1.744 & 0.231 \\
                        & $(5,400)$ & 0.079 & 0.082 & 0.134 & 0.117 & 0.078 &  & 15.629 & 15.549 & 0.218 & 0.137 & 0.102 \\
                        & $(5,800)$ & 0.057 & 0.058 & 0.112 & 0.044 & 0.048 &  & 23.770 & 23.742 & 0.247 & 0.050 & 0.081 \\
\multicolumn{1}{l}{}    &            & \multicolumn{1}{c}{} & \multicolumn{1}{c}{} & \multicolumn{1}{c}{} & \multicolumn{1}{c}{} & \multicolumn{1}{c}{} & \multicolumn{1}{c}{} & \multicolumn{1}{c}{} & \multicolumn{1}{c}{} & \multicolumn{1}{c}{} & \multicolumn{1}{c}{} & \multicolumn{1}{c}{} \\
\multirow{6}{*}{HAC}    & $(1,200)$  & 0.062 & 0.071 & 0.065 & 0.076 & 0.049 &  & 0.674  & 0.684  & 0.086 & 0.122 & 0.057 \\
                        & $(1,400)$  & 0.034 & 0.036 & 0.035 & 0.023 & 0.019 &  & 1.235  & 1.236  & 0.067 & 0.036 & 0.028 \\
                        & $(1,800)$  & 0.024 & 0.026 & 0.024 & 0.010 & 0.010 &  & 2.004  & 2.004  & 0.072 & 0.015 & 0.022 \\
                        & $(5,200)$  & 0.162 & 0.168 & 0.258 & 6.369 & 0.214 &  & 6.931  & 6.807  & 0.360 & 7.122 & 0.262 \\
                        & $(5,400)$ & 0.083 & 0.086 & 0.149 & 0.123 & 0.083 &  & 15.720 & 15.634 & 0.269 & 0.144 & 0.120 \\
                        & $(5,800)$ & 0.062 & 0.063 & 0.120 & 0.045 & 0.050 &  & 24.065 & 24.064 & 0.279 & 0.058 & 0.090           
 \\
\bottomrule
\end{tabular}
\end{adjustbox}
\caption{Mean Squared Errors under various methods}
\label{table:mse}
\end{table}

Table \ref{table:cover} reports coverage rates at a 5\% significance level for different estimation methods. Under the stationary model with $\lambda_t \sim N(1,1)$, all five methods achieve coverage rates of approximately 95\%. However, under the model with a time trend, i.e., $\lambda_t \sim N(\log(t),1)$, regression-based methods (\textbf{SC} and \textbf{SC-S}) produce a large bias and fail to cover the ground truth, while all three proximal inference methods (\textbf{PI}, \textbf{PI-P}, and \textbf{PI-S}) exhibit good coverage rates. The comparison of coverage rates in the "robust" section versus the "HAC" section demonstrates the validity of our HAC robust inference approach for models with time series autocorrelation. When the dimension of latent factors is large relative to the number of observations, the proximal inference methods with surrogates (\textbf{PI-P} and \textbf{PI-S}) tend to become conservative. Importantly, alongside Table \ref{table:mse}, we demonstrate the potential of applying proximal inference to estimate causal effects with post-intervention data only by highlighting relatively low MSEs and good coverage rates. Finally, Table \ref{table:mse_cov} in the Appendix showcases the benefits of incorporating measured covariates.

\begin{table}[ht!]
\centering
\setlength{\tabcolsep}{3pt}
\begin{adjustbox}{max width=0.95\linewidth,center}
\begin{tabular}{lllllllllllll}
\toprule
       &              & \multicolumn{5}{c}{$\lambda_t \sim N(1,1)$}                      &  & \multicolumn{5}{c}{$\lambda_t \sim N(\log(t),1)$}                      \\  \cmidrule{3-7} \cmidrule{9-13}
Std Error  & ($K,T$)  &  \textbf{SC}   & \textbf{SC-S}   & \textbf{PI} & \textbf{PI-P}  & \textbf{PI-S}   & & \textbf{SC}   & \textbf{SC-S}   & \textbf{PI} & \textbf{PI-P}  & \textbf{PI-S}   \\
\toprule
\multirow{6}{*}{Robust} & $(1,200)$  & 94.00\% & 91.85\% & 95.05\% & 93.25\% & 94.25\% &  & 11.40\% & 12.75\% & 94.55\% & 94.40\% & 94.25\% \\
                        & $(1,300)$  & 94.65\% & 93.35\% & 94.35\% & 94.50\% & 94.55\% &  & 0.00\%  & 0.00\%  & 94.40\% & 94.95\% & 94.40\% \\
                        & $(1,800)$  & 94.60\% & 94.00\% & 94.45\% & 94.75\% & 94.75\% &  & 0.00\%  & 0.00\%  & 94.30\% & 94.35\% & 95.00\% \\
                        & $(5,200)$  & 94.35\% & 94.30\% & 96.70\% & 97.65\% & 97.05\% &  & 0.15\%  & 0.15\%  & 95.95\% & 97.60\% & 97.60\% \\
                        & $(5,400)$ & 94.40\% & 93.95\% & 95.90\% & 95.75\% & 96.35\% &  & 0.00\%  & 0.00\%  & 95.45\% & 95.45\% & 96.00\% \\
                        & $(5,800)$ & 95.05\% & 94.65\% & 94.35\% & 95.35\% & 94.60\% &  & 0.00\%  & 0.00\%  & 94.30\% & 95.55\% & 94.80\% \\
\multicolumn{1}{l}{}    &            & \multicolumn{1}{c}{} & \multicolumn{1}{c}{} & \multicolumn{1}{c}{} & \multicolumn{1}{c}{} & \multicolumn{1}{c}{} & \multicolumn{1}{c}{} & \multicolumn{1}{c}{} & \multicolumn{1}{c}{} & \multicolumn{1}{c}{} & \multicolumn{1}{c}{} & \multicolumn{1}{c}{} \\
\multirow{6}{*}{HAC}    & $(1,200)$  & 93.30\% & 91.05\% & 93.85\% & 92.90\% & 93.80\% &  & 14.70\% & 15.35\% & 94.10\% & 93.10\% & 93.25\% \\
                        & $(1,400)$  & 92.95\% & 92.20\% & 93.95\% & 93.95\% & 94.55\% &  & 0.05\%  & 0.10\%  & 93.40\% & 93.95\% & 93.75\% \\
                        & $(1,800)$  & 94.05\% & 93.00\% & 94.65\% & 94.25\% & 95.30\% &  & 0.00\%  & 0.00\%  & 92.45\% & 94.95\% & 93.50\% \\
                        & $(5,200)$  & 93.60\% & 92.80\% & 95.45\% & 97.95\% & 96.40\% &  & 0.45\%  & 0.45\%  & 93.75\% & 97.30\% & 96.20\% \\
                        & $(5,400)$ & 93.75\% & 93.45\% & 94.70\% & 95.15\% & 96.15\% &  & 0.00\%  & 0.00\%  & 93.55\% & 95.45\% & 95.10\% \\
                        & $(5,800)$ & 94.85\% & 94.85\% & 94.40\% & 95.65\% & 95.40\% &  & 0.00\%  & 0.00\%  & 92.95\% & 94.45\% & 94.35\%          
 \\
\bottomrule
\end{tabular}
\end{adjustbox}
\caption{Coverage rate under various methods}
\label{table:cover}
\end{table}

\section{Empirical Application}\label{sec:empirical}
In this section, we examine and compare the five synthetic control methods discussed in Section \ref{sec:simulation} (\textbf{SC}, \textbf{SC-S}, \textbf{PI}, \textbf{PI-S}, and \textbf{PI-P}) using data collected from the study of the Panic of 1907, one of the most severe financial crises in US history. \cite{Fohlin2021} re-analyzed the impact of this event on NYC trusts using a new high-frequency dataset on market valuations and a generalized synthetic control method proposed by \cite{xu2017}. Their analysis employs a dataset covering 59 trust companies from January 5, 1906, to December 30, 1908, with a triweekly frequency, derived from the "Trust and Surety Company" table published in the New York Tribune. These trust companies are divided into three groups based on their connections and status during the crisis: the "troubled" group, the "connected" group, and the "independent" group. The troubled group comprises three trusts that experienced severe runs during the panic, while the connected group consists of seven trusts linked to four major "money trusts." The remaining 49 trusts are classified as the independent group.

In the original analysis by \cite{Fohlin2021}, the primary focus was on assessing the impact of the panic on both troubled and connected trusts collectively, taking into account the average performance of firms within each respective group. However, our study aims to highlight the disparities in estimation outcomes arising from the application of various methods explored in this paper. To achieve this, we consider the logarithm of the stock mid-price of Knickerbocker Trust Company as the target outcome, with the three troubled trusts serving as surrogates. Specifically, we use the bid price (logarithm) of these trusts, including the bid price (logarithm) of Knickerbocker as surrogates, with their corresponding ask prices acting as proxies. Similarly, the mid-prices (logarithm) of the 48 independent trusts are used as donors. \footnote{The original analysis consists of 49 independent trusts, one of which has a missing date and thus was dropped in our analysis. We adopt the contaminated surrogate approach from Section \ref{sec:contaminate}, as it is arguably more suitable for this dataset.}

\begin{figure}[ht!]
    \centering
    \includegraphics[width=0.7\textwidth]{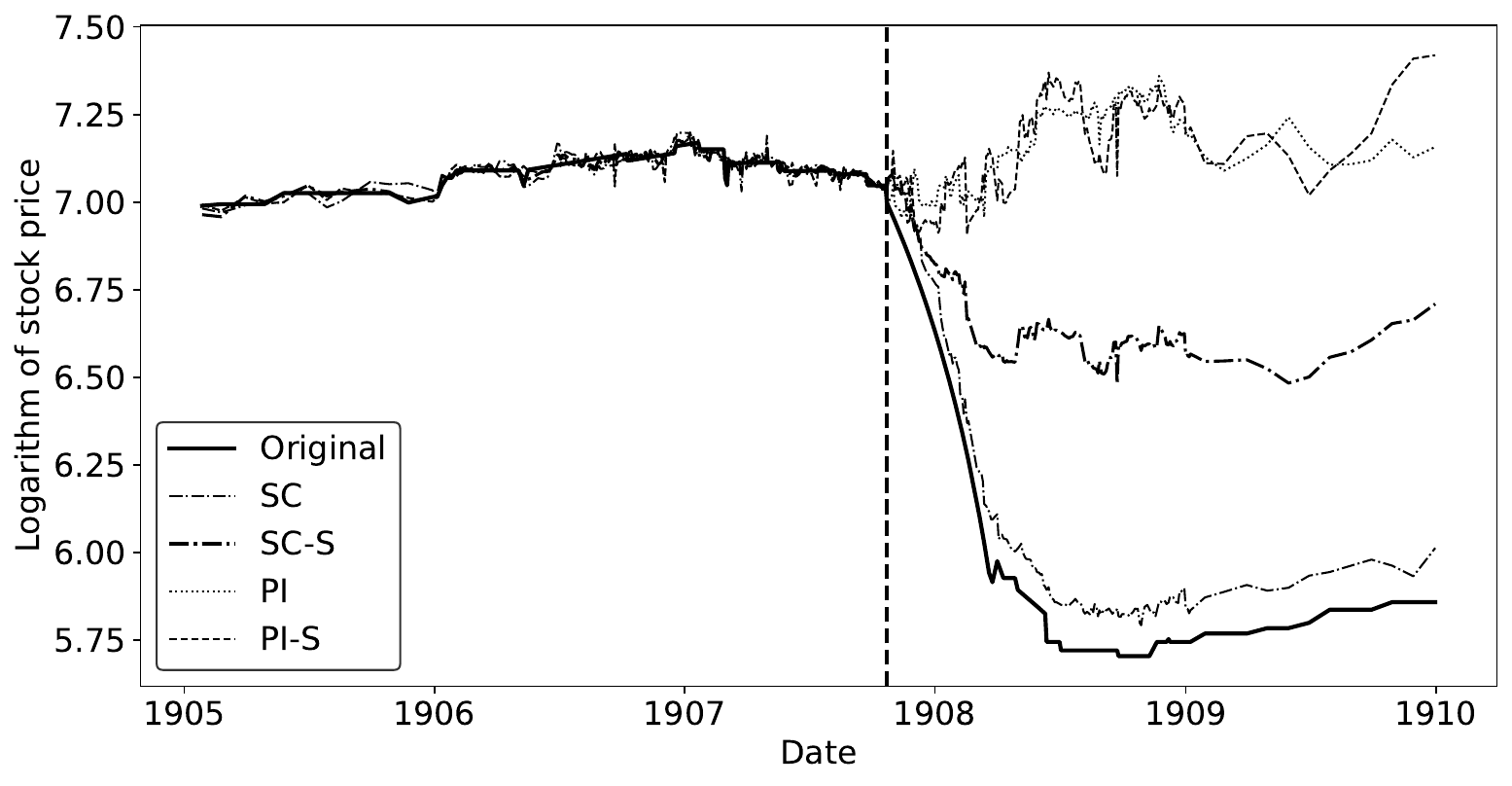}
    \caption{Stock price of Knickerbocker and its synthetic controls using various methods}
    \label{fig:empirical}
\end{figure}

Figure \ref{fig:empirical} displays the trends of the logarithmic stock price of Knickerbocker Trust Company alongside the synthetic control estimates derived from the five methods presented in the Section \ref{sec:simulation}. The plot reveals that the classic synthetic control regression generates a counterfactual series closely aligned with the observed outcome in the post-intervention period, resulting in a relatively small causal effect. In contrast, the proximal inference methods (\textbf{PI} and \textbf{PI-S}) yield series that visually extend the pre-intervention price trend more consistently, suggesting that these approaches might be better suited to this context. 
Notably, the post-intervention predictions, particularly with \textbf{PI} and \textbf{PI-S}, exhibit higher volatility compared to the pre-intervention series, reflecting the inherent uncertainty associated with the estimates. Meanwhile, \textbf{SC-S} remains below the pre-intervention average price but also appears to perform relatively well.

\begin{table}[!h]
\centering
\setlength{\tabcolsep}{12pt}
\begin{threeparttable}
\caption{Estimated ATT under different lengths of post-intervention period}
\begin{tabular}{ccccccc}
\toprule
\multicolumn{1}{l}{\#time steps} & $\text{SC}^{\dagger}$      & $\text{SC-S}^{\dagger}$    & $\text{PI}^{\dagger}$      & $\text{PI-S}^{*}$    & $\text{PI-P}^{\dagger}$    \\
\midrule
\multirow{2}{*}{80}                 & -0.020 & -0.370 & -0.600 & -0.593 & -0.739 \\
                                   & (120.843)  & (12.750)  & (0.999)  & (\textbf{0.154})  & (0.022)  \\ 
\multirow{2}{*}{100}                & -0.013 & -0.495 & -0.771 & -0.769 & -0.361 \\
                                   & (182.137)  & (9.725)  & (0.972)  & (\textbf{0.161})  & (0.012)  \\ 
\multirow{2}{*}{120}                & -0.061 & -0.566 & -0.909 & -0.920 & -0.590 \\
                                   & (113.319)  & (10.030)  & (0.998)  & (\textbf{0.165})  & (0.245)  \\ 
\multirow{2}{*}{160}                & -0.069 & -0.626 & -1.086 & -1.086 & -0.531 \\
                                   & (107.123)  & (10.882)  & (0.959)  & (\textbf{0.166})  & (0.082)  \\ 
\multirow{2}{*}{182}               & -0.089 & -0.643 & -1.138 & -1.134 & -1.220 \\
                                   & (56.464)  & (11.036)  & (0.962)  & (\textbf{0.158})   & (0.092) 
\\\bottomrule
\end{tabular}
\begin{tablenotes}
\scriptsize
\item[*] Numbers in parentheses represent the estimated standard errors.
\item[\dag] Numbers in parentheses represent the relative efficiencies, calculated as the ratio of variances of PI-S against the other methods.
\end{tablenotes}
\end{threeparttable}
\label{table:empirical}
\end{table}

Table \ref{table:empirical} displays the estimated average treatment effect on the treated (ATT) obtained using the five methods, along with their corresponding relative efficiency and standard errors enclosed in parentheses. In the \textbf{PI-S} column, the standard errors of the estimator are included, while in the other columns, we show the relative efficiency calculated as the ratio of variances of \textbf{PI-S} against the other methods. To estimate standard errors and relative efficiency, we use heteroskedasticity- and autocorrelation-consistent (HAC) variance estimators. Each row section corresponds to different lengths of post-intervention periods used for estimation, with the total length of the post-treatment period available in the data being 182 (triweekly frequency). The results suggest that the \textbf{SC} method performs poorly, as its point estimates differ considerably from those of the other methods. On the other hand, the \textbf{PI-S} method demonstrates smaller standard errors than \textbf{PI} when employing different post-intervention periods for training. This finding highlights the benefits of augmenting pre-intervention estimation with post-intervention data. Nonetheless, the difference between \textbf{PI} and \textbf{PI-S} is not substantial. Interestingly, despite relying solely on post-intervention data, \textbf{PI-P} is capable of producing reasonably accurate estimates. This observation underscores the potential advantages of using proximal inference methods, such as \textbf{PI}, \textbf{PI-P} and \textbf{PI-S}, in estimating the impact of financial crises on stock prices.

\section{Conclusion}\label{sec:conclusion}
Our paper introduces a novel framework that enhances the synthetic control method (SCM) for estimating the long-term impact of interventions. By integrating surrogates and leveraging the proximal causal inference framework, our approach effectively incorporates post-intervention information, improving the estimation of causal effects. Importantly, our method demonstrates the feasibility of estimating causal effects using only post-intervention period data, which is particularly valuable in situations with limited pre-intervention data. Moving forward, several avenues for future research arise from our work. One direction involves the development of prediction intervals for our proposed estimator, as seen in the work of \cite{Chernozhukov2021}. This would enable the quantification of uncertainty and provide a more comprehensive assessment of the estimated treatment effects. Additionally, exploring weighting approaches to achieve double robustness properties, as discussed in \cite{qiu2023doubly}, could further enhance the robustness of our estimates. Furthermore, in Section \ref{sec:contaminate}, we present a solution for surrogates contaminated by the outcomes of the control group. Extending this approach to address interference problems in the synthetic control setup holds promise for future research. Overall, our framework opens up new possibilities for refining and expanding the synthetic control method, offering valuable insights for researchers in the field of impact evaluation. 

\newpage
\bibliography{reference}

\newpage
\appendix
\section{Proof of Theorem \ref{thm:identification}}\label{sec:proof}
By Assumption \ref{as:error-term} and \ref{as:existence-of-weight}, we have
\begin{align*}
    E[Y_t(0)] = E[\lambda_t^\prime \beta] = E[\lambda_t^\prime \Gamma \alpha ] = E[(W_t - \epsilon_{W,t})^\prime \alpha] = E[W_t^\prime \alpha]
\end{align*}
and
\begin{equation*}
    E[Y_t(1)] = E[Y_t(0)] + E[\rho_t^\prime \theta] = E[W_t^\prime \alpha] + E[\rho_t^\prime \Phi\gamma] = E[W_t^\prime \alpha] + E[(X_t - \epsilon_{X,t})^\prime\gamma] = E[W_t^\prime \alpha + X_t^\prime \gamma]~.
\end{equation*}
Again, by Assumption \ref{as:error-term} and \ref{as:existence-of-weight}, we have
\begin{align*}
    E[Y_t(0) - W_t^\prime \alpha \mid \lambda_t, t\leq T_0] = E[\epsilon_{Y,t} - \epsilon_{W,t}^\prime \alpha \mid \lambda_t, t<T_0] = 0
\end{align*}
and
\begin{align*}
    E[Y_t(1) - W_t^\prime \alpha - X_t^\prime \gamma \mid \lambda_t, \rho_t, t > T_0] &= E[Y_t(1)-Y_t(0) - X_t^\prime \gamma \mid \lambda_t, \rho_t, t > T_0] \\
    &= E[\delta_t - \epsilon_{X,t}^\prime \gamma \mid \lambda_t, \rho_t, t > T_0] = 0~.
\end{align*}
By Assumption \ref{as:existence-of-proxy}, we have
\begin{align*}
    E[Y_t - W_t^\prime \alpha \mid Z_{0,t}, t\leq T_0] &= E[ E[Y_t(0) - W_t^\prime \alpha \mid \lambda_t, Z_{0,t}, t\leq T_0] \mid Z_{0,t}, t\leq T_0 ] \\
    &= E[ E[Y_t(0) - W_t^\prime \alpha \mid \lambda_t, t\leq T_0] \mid Z_{0,t}, t\leq T_0 ] = 0~,
\end{align*}
and
\begin{align*}
    E[Y_t - W_t^\prime \alpha - X_t^\prime \gamma \mid Z_{0,t}, Z_{1,t}, t > T_0]  &= E[E[Y_t(1) - W_t^\prime \alpha - X_t^\prime \gamma \mid \lambda_t, \rho_t, Z_{1,t}, t > T_0] \mid Z_{0,t}, Z_{1,t}, t > T_0] \\
    &=E[E[Y_t(1) - W_t^\prime \alpha - X_t^\prime \gamma \mid \lambda_t, \rho_t, t > T_0] \mid Z_{0,t}, Z_{1,t}, t > T_0] = 0~.
\end{align*}
To show uniqueness, assume there exists two sets of vectors of weights $(\alpha, \gamma), (\Tilde{\alpha},\Tilde{\gamma})$ such that $E[Y_t - W_t^\prime \alpha - X_t^\prime \gamma \mid Z_{0,t}, Z_{1,t}, t > T_0] = E[Y_t - W_t^\prime \Tilde{\alpha} - X_t^\prime \Tilde{\gamma} \mid Z_{0,t}, Z_{1,t}, t > T_0] = 0$. Thus, $E[W_t^\prime (\alpha - \Tilde{\alpha}) + X_t^\prime (\gamma - \Tilde{\gamma}) \mid Z_{0,t},  Z_{1,t}, t > T_0]=0$, which by Assumption \ref{as:completeness} implies $\alpha = \Tilde{\alpha}$ and $\gamma = \Tilde{\gamma}$. 
\qed

\section{Proof of Theorem \ref{thm:gmm-coefficients} and \ref{thm:gmm-ate}}\label{sec:proof2}
We first show that $E[U_t(\alpha, \gamma)] = E[\Tilde{U}_t(\alpha, \gamma, \tau)] = 0$ at the true value under Assumptions \ref{as:error-term}-\ref{as:completeness} and \ref{as:stationary}. For brevity, we only show $E[U_t(\alpha, \gamma)] = 0$. We have
\begin{align*}
    E[U_t(\alpha, \gamma)] &= E\left[\begin{pmatrix}
    g_0(Z_{0,t}) (Y_t - W_t^\prime \alpha) I\{t \leq T_0\} \\
    g_1(Z_{0,t}, Z_{1,t}) (Y_t - W_t^\prime \alpha - X_t^\prime \gamma) I\{t > T_0\}
    \end{pmatrix}\right] \\
    &= E\left[\begin{pmatrix}
    E\left[g_0(Z_{0,t}) (Y_t - W_t^\prime \alpha) I\{t \leq T_0\}\mid Z_{0,t}, t \leq T_0\right] \\
    E\left[g_1(Z_{0,t}, Z_{1,t}) (Y_t - W_t^\prime \alpha - X_t^\prime \gamma) I\{t > T_0\} \mid Z_{0,t}, Z_{1,t}, t > T_0\right]
    \end{pmatrix}\right] \\
    &= E\left[\begin{pmatrix}
    g_0(Z_{0,t})E\left[ (Y_t - W_t^\prime \alpha) I\{t \leq T_0\}\mid Z_{0,t}, t \leq T_0\right] \\
    g_1(Z_{0,t}, Z_{1,t}) E\left[(Y_t - W_t^\prime \alpha - X_t^\prime \gamma) I\{t > T_0\} \mid Z_{0,t}, Z_{1,t}, t > T_0\right]
    \end{pmatrix}\right] \quad \text{(by Theorem \ref{thm:identification})}\\
    &= 0~.
\end{align*}
Now we focus on proving Theorem \ref{thm:gmm-coefficients}. Similar arguments apply to Theorem \ref{thm:gmm-ate}. We follow Chapter 3 of \cite{hall2005generalized} to prove the asymptotic normality of our estimator. To summarize, we have a population moment condition $E[U_t(\theta)] = 0$ with a moment function $U_t(\alpha, \gamma) = \mathcal{V}_t\left(Y_t-\mathcal{D}_t^{\prime} \theta\right) $, $\theta = (\alpha, \gamma) \in \mathcal{R}^{N+H}$, $\mathcal{V}_t = \{g_0(Z_{0,t})^\prime I\{t \leq T_0\}, g_1(Z_{0,t}, Z_{1,t})^\prime I\{t > T_0\}\}^\prime \in \mathcal{R}^{d_1 + d_2}$, $\mathcal{D}_t = \{W_t^\prime, X_t^\prime I\{t > T_0\} \}^\prime \in \mathcal{R}^{N+H}$. Let $\mathcal{O}_t = \{Y_t, W_t, X_t, T_0, Z_{0,t}, Z_{1,t}\}\in \mathcal{R}^{N+H+\text{dim}(\{Z_{0,t},Z_{1,t}\})+2}$, denote the observable vector of random variables, where $\text{dim}(\{Z_{0,t},Z_{1,t}\})$ is the dimension of $\{Z_{0,t},Z_{1,t}\}$. Let $\Theta \subseteq \mathcal{R}^{N+H}$ denote the parameter space of $\theta$, and let $O \subseteq \mathcal{R}^{N+H+\text{dim}(\{Z_{0,t},Z_{1,t}\})+2}$ denote the sample space of $\mathcal{O}_t$. Then $U_t = U_t(\mathcal{O}_t; \theta)$ is a mapping from $O\times \Theta$ to $\mathcal{R}^{d_1 + d_2}$. We impose the following regularity conditions.

\begin{assumption}[Strict stationarity]\label{as:app-stationary}
    The observable vector of random variables $\mathcal O_t$ form a strictly stationary process, such that all expectations of functions of $\mathcal O_t$ do not depend on time.
\end{assumption}

\begin{assumption}[Regularity conditions for $U_t$]\label{as:app-ut}
    The function $U_t: O\times \Theta \rightarrow \mathcal{R}^{d_1 + d_2}$ where $d_1 + d_2 < \infty$ satisfies: (i) it is continuous on $\Theta$ for each $\mathcal{O}_t \in O$; (ii) $E[U_t(\mathcal{O}_t;\theta)]$ exists and is finite for every $\theta \in \Theta$; (iii) $E[U_t(\mathcal{O}_t;\theta)]$ is continuous on $\Theta$.
\end{assumption}

\begin{assumption}[Regularity conditions on $\partial U_t(\mathcal{O}_t;\theta)/\partial \theta^\prime$]\label{as:app-regularity-Ut}
    (i) The derivative matrix $\partial U_t(\mathcal{O}_t;\theta)/\partial \theta^\prime$ exists and is continuous on $\Theta$ for each $\mathcal{O}_t \in O$; (ii) The true value of $\theta$ does not lie on the boundary of $\Theta$; (iii) $E[\partial U_t(\mathcal{O}_t;\theta)/\partial \theta^\prime]$ exists and is finite.
\end{assumption}

\begin{assumption}[Properties of the Weighting Matrix]\label{as:app-omega}
    The user-specified weight matrix $\Omega$ is a positive semi-definite matrix, possibly depends on data, and converges in probability to the positive definite matrix of constants.
\end{assumption}

\begin{assumption}[Ergodicity]\label{as:app-ergodicity}
    The random process $\{\mathcal{O}_t; -\infty < t < \infty\}$ is ergodic.
\end{assumption}

\begin{assumption}[Compactness of $\Theta$]\label{as:app-compact-Theta}
    $\Theta$ is a compact set.
\end{assumption}

\begin{assumption}[Domination of $U_t(\mathcal{O}_t;\theta)$]\label{as:app-domination-U}
    $E[\text{sup}_{\theta\in \Theta} \|U_t(\mathcal{O}_t;\theta)\|] < \infty$.
\end{assumption}

\begin{assumption}[Properties of the variance of the sample moment]\label{as:app-variance}
    Let $\theta^* = (\alpha^*, \gamma^*)$ denote the true value of $\theta$. (i) $E[U_t(\mathcal{O}_t;\theta) U_t(\mathcal{O}_t;\theta)^\prime]$ exists and is finite; (ii) $S$ exists and is a finite valued positive definite matrix.
\end{assumption}

\begin{assumption}[Properties of $G_T(\theta) = T^{-1} \sum_{t=1}^T \partial U_t(\mathcal{O}_t;\theta)/\partial \theta^\prime$]\label{as:app-Gt}
    (i) $E[\partial U_t(\mathcal{O}_t;\theta)/\partial \theta^\prime]$ is continuous on some neighbourhood $N_\epsilon$ of the true value $\theta^*$ in $\Theta$; (ii) Uniform convergence of $G_T(\theta)$: $\sum_{\theta \in N_\epsilon} \|G_T(\theta) -  E[\partial U_t(\mathcal{O}_t;\theta)/\partial \theta^\prime]\| \xrightarrow{p} 0$.
\end{assumption}

Under Assumptions \ref{as:error-term}-\ref{as:completeness} and \ref{as:stationary}  and \ref{as:app-stationary}-\ref{as:app-Gt}, we have that Theorem \ref{thm:gmm-coefficients} holds by Theorem 3.2 of \cite{hall2005generalized}.

\section{Additional Tables}
\begin{table}[ht!]
\centering
\setlength{\tabcolsep}{5.5pt}
\begin{adjustbox}{max width=0.95\linewidth,center}
\begin{tabular}{lllllllllllll}
\toprule
       &              & \multicolumn{5}{c}{$\lambda_t \sim N(1,1)$}                      &  & \multicolumn{5}{c}{$\lambda_t \sim N(\log(t),1)$}                      \\  \cmidrule{3-7} \cmidrule{9-13}
Error term & ($K,T$)  &  \textbf{SC}   & \textbf{SC-S}   & \textbf{PI} & \textbf{PI-P}  & \textbf{PI-S}   & & \textbf{SC}   & \textbf{SC-S}   & \textbf{PI} & \textbf{PI-P}  & \textbf{PI-S}   \\
\toprule
\multirow{6}{*}{Covariates} & $(1,400)$  & 0.031 & 0.031 & 0.033 & 0.036  & 0.025 &  & 0.662 & 0.662 & 0.044 & 0.055 & 0.026 \\
                        & $(1,800)$  & 0.015 & 0.015 & 0.016 & 0.016  & 0.012 &  & 0.639 & 0.639 & 0.021 & 0.026 & 0.013 \\
                        & $(1,1200)$  & 0.010 & 0.010 & 0.011 & 0.011  & 0.008 &  & 0.635 & 0.634 & 0.014 & 0.019 & 0.009 \\
                        & $(5,400)$  & 0.081 & 0.081 & 0.112 & 0.208  & 0.089 &  & 6.674 & 6.611 & 0.144 & 0.264 & 0.108 \\
                        & $(5,800)$ & 0.037 & 0.037 & 0.051 & 0.083  & 0.042 &  & 6.463 & 6.430 & 0.066 & 0.109 & 0.049 \\
                        & $(5,1200)$ & 0.026 & 0.026 & 0.036 & 0.057  & 0.030 &  & 6.394 & 6.365 & 0.043 & 0.065 & 0.032 \\
\multicolumn{1}{l}{}    &            & \multicolumn{1}{c}{} & \multicolumn{1}{c}{} & \multicolumn{1}{c}{} & \multicolumn{1}{c}{} & \multicolumn{1}{c}{} & \multicolumn{1}{c}{} & \multicolumn{1}{c}{} & \multicolumn{1}{c}{} & \multicolumn{1}{c}{} & \multicolumn{1}{c}{} & \multicolumn{1}{c}{} \\
\multirow{6}{*}{No Covariates}    & $(1,400)$  & 0.043 & 0.043 & 0.060 & 0.080  & 0.049 &  & 0.677 & 0.677 & 0.070 & 0.144 & 0.050 \\
                        & $(1,800)$  & 0.020 & 0.020 & 0.027 & 0.038  & 0.024 &  & 0.647 & 0.647 & 0.036 & 0.062 & 0.026 \\
                        & $(1,1200)$  & 0.013 & 0.013 & 0.017 & 0.025  & 0.016 &  & 0.647 & 0.647 & 0.024 & 0.042 & 0.018 \\
                        & $(5,400)$  & 0.093 & 0.092 & 0.214 & 15.755 & 0.226 &  & 6.650 & 6.588 & 0.248 & 1.582 & 0.238 \\
                        & $(5,800)$ & 0.044 & 0.044 & 0.090 & 0.228  & 0.089 &  & 6.441 & 6.411 & 0.116 & 0.315 & 0.102 \\
                        & $(5,1200)$ & 0.029 & 0.029 & 0.058 & 0.150  & 0.061 &  & 6.432 & 6.406 & 0.072 & 0.173 & 0.064 
 \\
\bottomrule
\end{tabular}
\end{adjustbox}
\caption{Mean Squared Errors under various methods}
\label{table:mse_cov}
\end{table}

\end{document}